\documentclass{imsart}

\RequirePackage{amsthm,amsmath,amsfonts,amssymb}
\RequirePackage[numbers]{natbib}

\startlocaldefs
\newtheorem{theorem}{Theorem}

\theoremstyle{definition}

\newtheorem{example}{Example}

\usepackage{mathtools} 
\usepackage{xcolor}
\usepackage[inline]{enumitem}
\usepackage{multirow,arydshln} 

\newcommand{\mypar}[1]{\smallskip\noindent\textbf{#1.}}

\newcommand{\Prob}[1][\R^d]{\mathcal P\left({#1}\right)}  
\newcommand{\PP}{\mathsf P}								
\newcommand{\R}{\mathbb R}								
\newcommand{\dd}{\,\mathrm{d}\,}						
\newcommand{\tr}{^\mathsf{T}}							

\newcommand{\FS}{\mathcal F}    
\newcommand{\Rinf}{\mathbb R^*} 
\newcommand{\band}[1]{\mathcal B\left( #1 \right)} 
\newcommand{\BD}{BD} 
\newcommand{\MBD}{MBD} 
\newcommand{\Du}{D_1} 
\newcommand{\DuH}{D_{1,H}} 
\newcommand{\DuS}{D_{1,S}} 
\DeclarePairedDelimiter\ceil{\lceil}{\rceil}

\newcommand{\aff}{\mathcal A}   
\newcommand{\calX}{\mathcal X}  
\newcommand{\median}[1]{\mathrm{med}\left(#1\right)} 
\newcommand{\conv}[1]{\mathrm{conv}\left(#1\right)} 

\newcommand{\pkg}[1]{{\normalfont\fontseries{b}\selectfont #1}}
\let\proglang=\textsf
\let\code=\texttt

\endlocaldefs

\begin{document}

\begin{frontmatter}
\title{Which depth to use to construct functional boxplots?}
\runtitle{Which depth to use to construct functional boxplots?}

\begin{aug}
\author[A]{\fnms{Stanislav}~\snm{Nagy}\ead[label=e1]{nagy@karlin.mff.cuni.cz}},
\author[B]{\fnms{Tom\'{a}\v{s}}~\snm{Mrkvi\v{c}ka}\ead[label=e2]{mrkvicka.toma@gmail.com}}
\and
\author[C]{\fnms{Antonio}~\snm{El\'ias}\ead[label=e3]{aelias@uma.es}}
\runauthor{Nagy et al.}
\address[A]{Department of Probability and Mathematical Statistics, Faculty of Mathematics and Physics, Charles University, Prague, Czech Republic\printead[presep={\ }]{e1}.}
\address[B]{Department of Applied Mathematics and Informatics, Faculty of Economics, University of South Bohemia, \v{C}esk\'e Bud\v{e}jovice, Czech Republic\printead[presep={\ }]{e2}.}
\address[C]{OASYS Group, Department of Mathematical Analysis, Statistics and Operations Research and, Applied Mathematics Faculty of Sciences, University of Malaga, M\'alaga, Spain\printead[presep={\ }]{e3}.}
\end{aug}

\begin{abstract}
This paper answers the question of which functional depth to use to construct a boxplot for functional data. It shows that integrated depths, e.g., the popular modified band depth, do not result in well-defined boxplots. Instead, we argue that infimal depths are the only functional depths that provide a valid construction of a functional boxplot. We also show that the properties of the boxplot are completely determined by properties of the one-dimensional depth function used in defining the infimal depth for functional data. Our claims are supported by (i) a motivating example, (ii) theoretical results concerning the properties of the boxplot, and (iii) a simulation study. 
\end{abstract}


\end{frontmatter}

%
%
%

\section{Introduction: Functional boxplot}  \label{section: introduction}

The functional boxplot \citep{Sun_Genton2011} is a celebrated visualization tool for functional data.\footnote{Functional boxplot is implemented in \proglang{R} packages \pkg{fda} \citep[function~\code{fbplot}]{R_fda}, \pkg{depthProc} \citep[function~\code{fncBoxPlot}]{Kosiorowski_Zawadzki2019}, and \pkg{GET} \citep[function~\code{fBoxplot}]{Myllymaki_Mrkvicka2019}, and in \proglang{Python} modules \pkg{statsmodels} \citep[function~\code{fboxplot}]{Python_Statsmodels} and \pkg{scikit-fda} \citep[function~\code{Boxplot}]{Ramos_etal2023}.} Using the concept of functional depth, that tool extends the idea of the classical boxplot to the situation when the observations can be represented as real-valued functions defined on a common domain. The functional boxplot consists of three elements:
    \begin{itemize}
        \item the \emph{median function}, defined as the sample function with the highest depth, 
        \item the \emph{central region} (also called \emph{box}) given as the band\footnote{A \emph{band} of functions is the region in $\R^2$ enclosed by the graphs of its delimiting functions, for a rigorous definition see~\eqref{eq: band} in Section~\ref{section: band}.} of $50\%$ of the sample functions with highest depth values, and
        \item the \emph{whiskers band}, obtained by inflating the central region around the median function by a constant factor.
    \end{itemize}
The central region is an analog of the box from $\R$. The whiskers band is used to flag outlying functional data --- if a function crosses its boundary or lies outside, the function is considered to be an outlier. The original functional boxplot in~\cite{Sun_Genton2011} was constructed so that the whiskers band is obtained by inflating the central band by a factor of~$4$, which corresponds to adding $1.5$-times the width of the central region to both sides of that band. A smaller adjustment factor for detecting outliers in functional data was considered in \citet{Sun_Genton2012}; see also \citet{Dai_Genton2018} for a two-stage procedure for constructing a functional boxplot, and \citet{Genton_Sun2020} for a review of further extensions and modifications of the functional boxplot.

The functional boxplot as a concept can be constructed using any functional depth. Many functional depths have been proposed in the literature; for a partial overview of these, see \citet{Gijbels_Nagy2017}. In \citet{Sun_Genton2011}, either the classical band depth or the modified band depth (MBD) for functional data \citep{Lopez_Romo2009} are considered. The popular functional boxplot based on MBD is a boxplot of an integrated type; Figure~\ref{figure: motivation} shows an example.\footnote{In the boxplot in Figure~\ref{figure: motivation}, we depict the whiskers band. Another possibility is plotting the band of all sample functions $X_1, \dots, X_n$ contained inside the whiskers band, as done in \citet{Sun_Genton2011}.} Functional boxplots based on the (modified) band depth are not the only proposal found in the literature. Boxplots based on depths of the so-called infimal (or extremal) type were proposed in \citet[Section~5.1]{Narisetty_Nair2016}.\footnote{Formal definitions of integrated and infimal depths are given in Section~\ref{section: functional depth}.} 

\begin{figure}
    \centering
    \includegraphics[width=\linewidth]{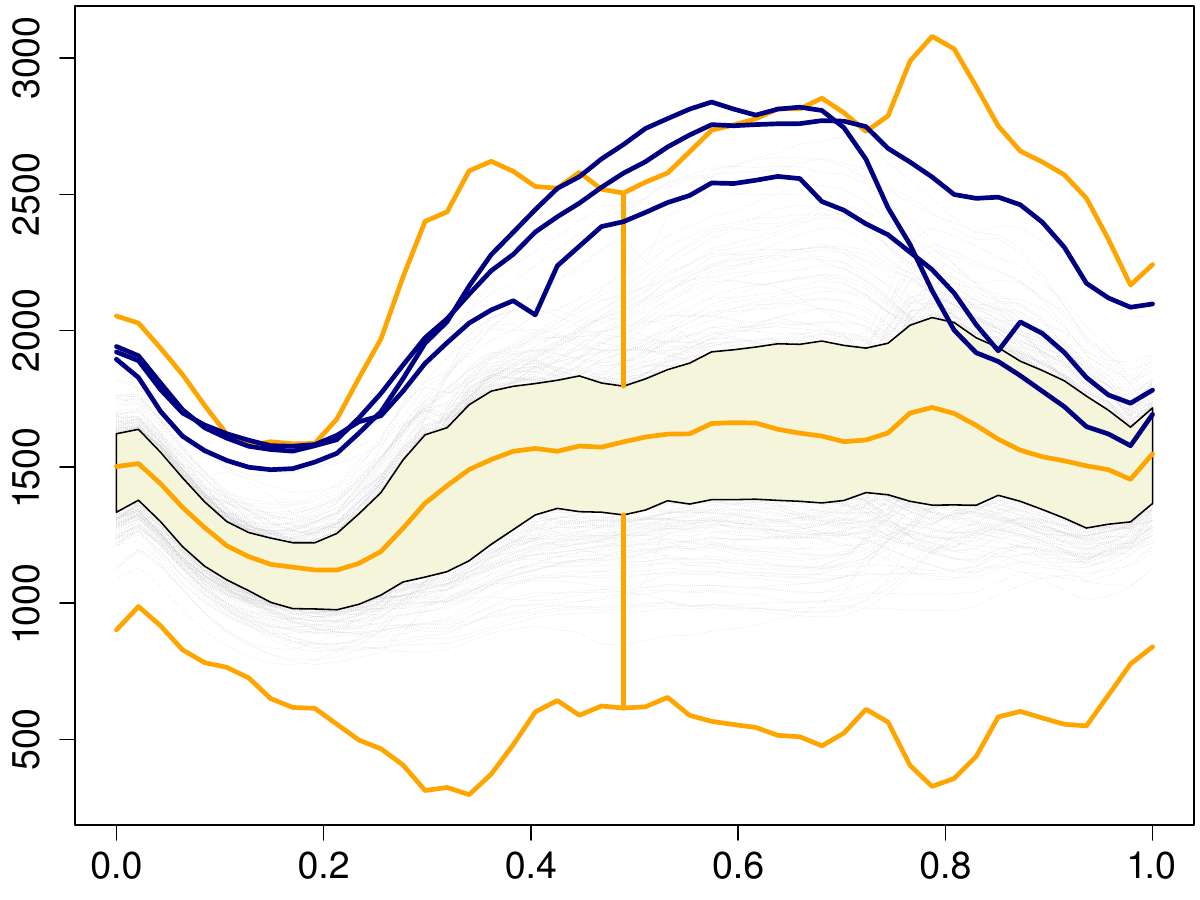}
    \caption{An MBD-based functional boxplot of the $508$ functions from the Monday electricity demand dataset \citep[dataset \code{Electricitydemand}]{R_fds}. The central region of $50\%$ deepest functions is the beige band. Inflating that band by a factor of $4$ around the median function (the thick orange curve in the middle of the $50\%$ central region) gives the whiskers (thick orange curves) that delimit the whiskers band of non-outlying functional data. Three functions are flagged as outliers (thick blue curves).}
    \label{figure: motivation}
\end{figure}

This work aims to point out that the boxplot based on an infimal depth has much better properties than a boxplot based on an integrated depth (e.g., MBD). The reason is that the boxplot uses local information to decide which function is an outlier --- namely, a function is flagged as outlying if it lies at any point outside the whiskers band. Thus, if the depth that defines the central region follows global principles, as all integrated depths do, then it does not match the locality paradigm used to flag outliers. The following motivating example demonstrates this. 

Assume a simple Gaussian model on the interval $[0,1]$ and two kinds of deviating functions. First is a \emph{local} outlier, which significantly departs from the base model on bounded sub-domains of $[0,1]$. Second is a \emph{global} outlier, which deviates from the base model in a less pronounced way, but on the whole domain $[0,1]$. The functions that we used for this illustrative example are shown in Figure~\ref{MotExf}. Figure~\ref{MotExFB} shows the functional boxplots using MBD (left) and the \textsf{erl} index (right) for our set of 120 functions, where 100 are from the base model, 10 are local outliers, and 10 are global outliers. The \textsf{erl} index is a simple tie-breaking refinement of the infimal depth commonly used in the literature \citep{Narisetty_Nair2016, Myllymaki_Mrkvicka2019}.

\begin{figure}
    \centering
    \includegraphics[width=\linewidth]{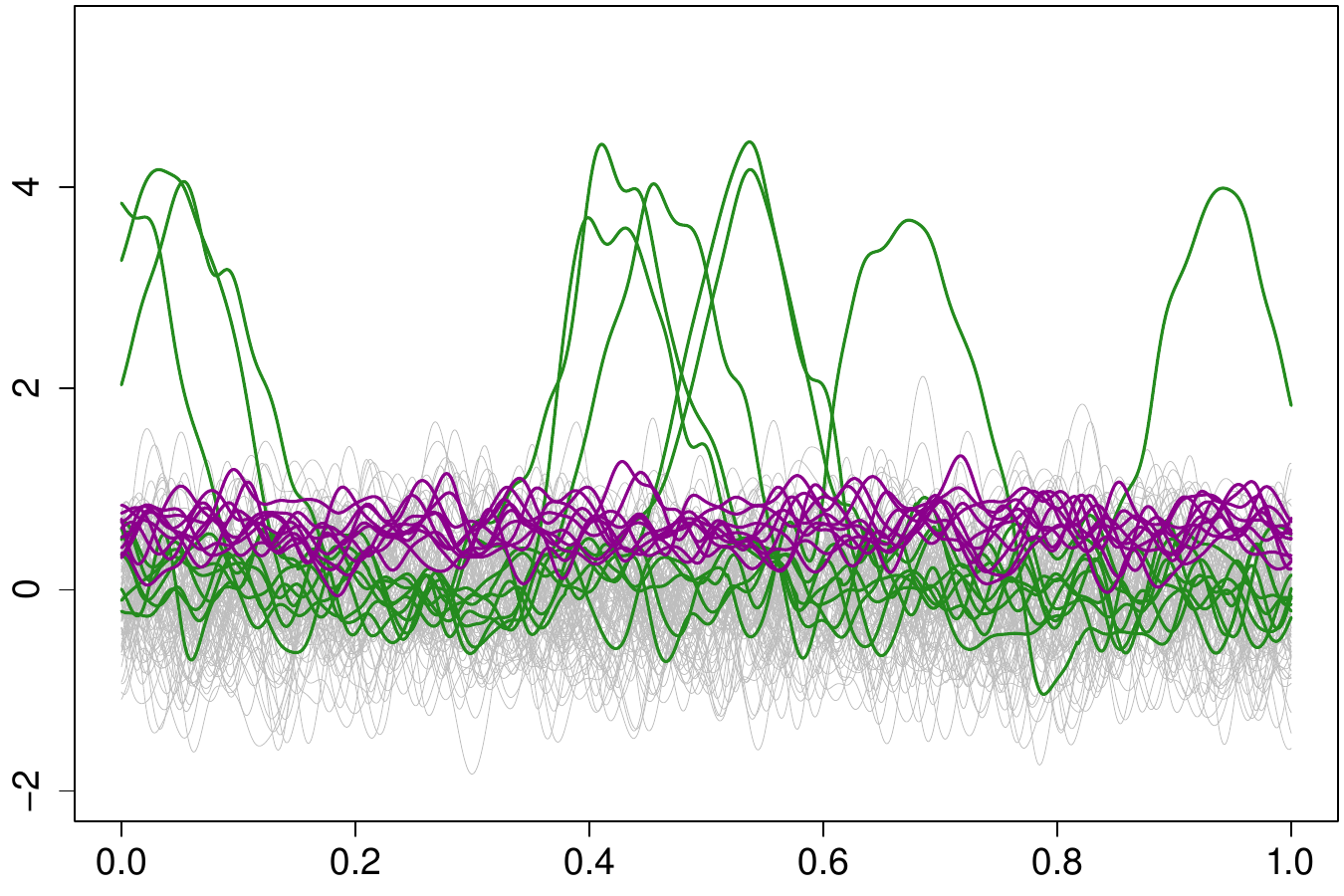}
    \caption{A simulated sample of $120$ functions: (i) 100 functions of the base model (gray), (ii) 10 functions that are local outliers with respect to (w.r.t.) the base model (green), and (iii) 10 functions that are global outliers w.r.t. the base model (magenta).}
    \label{MotExf}
\end{figure}

\begin{figure}
    \centering
    \includegraphics[width=.48\linewidth]{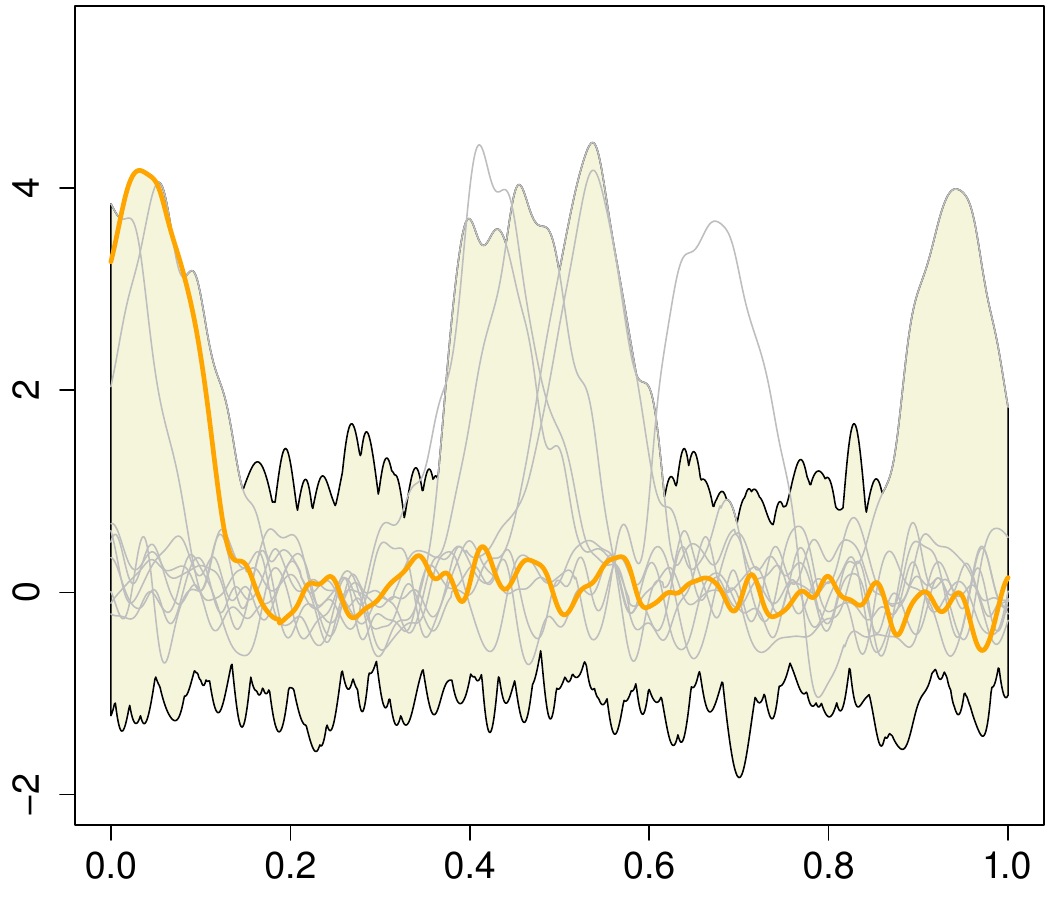} \hfill 
    \includegraphics[width=.48\linewidth]{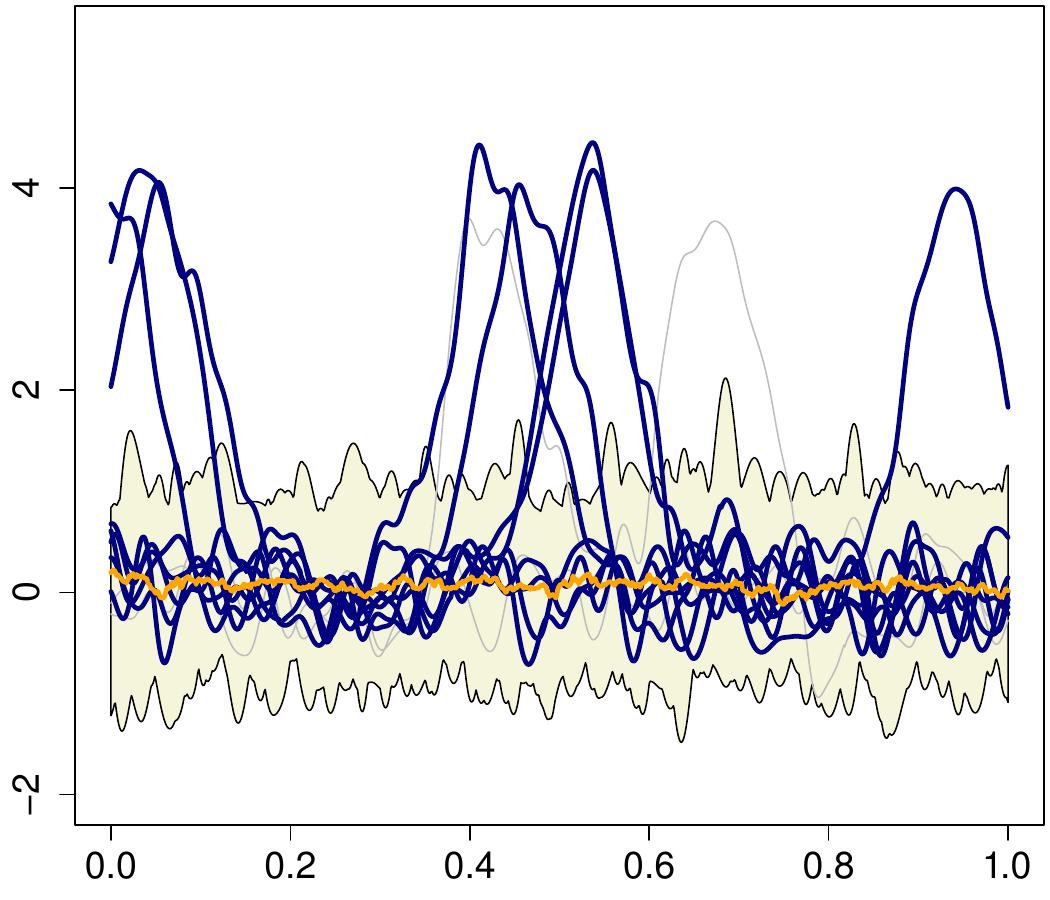}
    \caption{Functional boxplots for the motivating example based on MBD (left), and \textsf{erl} (right). The functional median is displayed as a thick orange curve; the band of $50\%$ of the most central functions is the beige region, and the detected outliers are in thick blue. No outliers are detected using the boxplot based on MBD.}
    \label{MotExFB}
\end{figure}

\begin{figure}
    \centering
    \includegraphics[width=.45\linewidth]{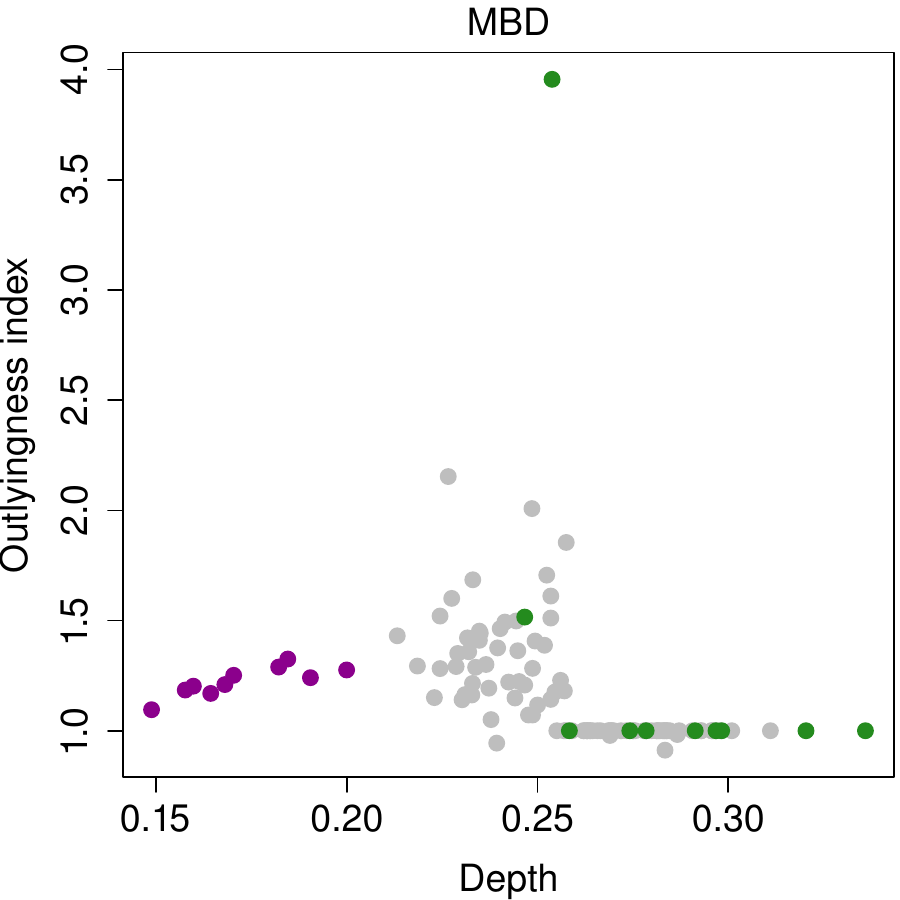} \hfill 
    \includegraphics[width=.45\linewidth]{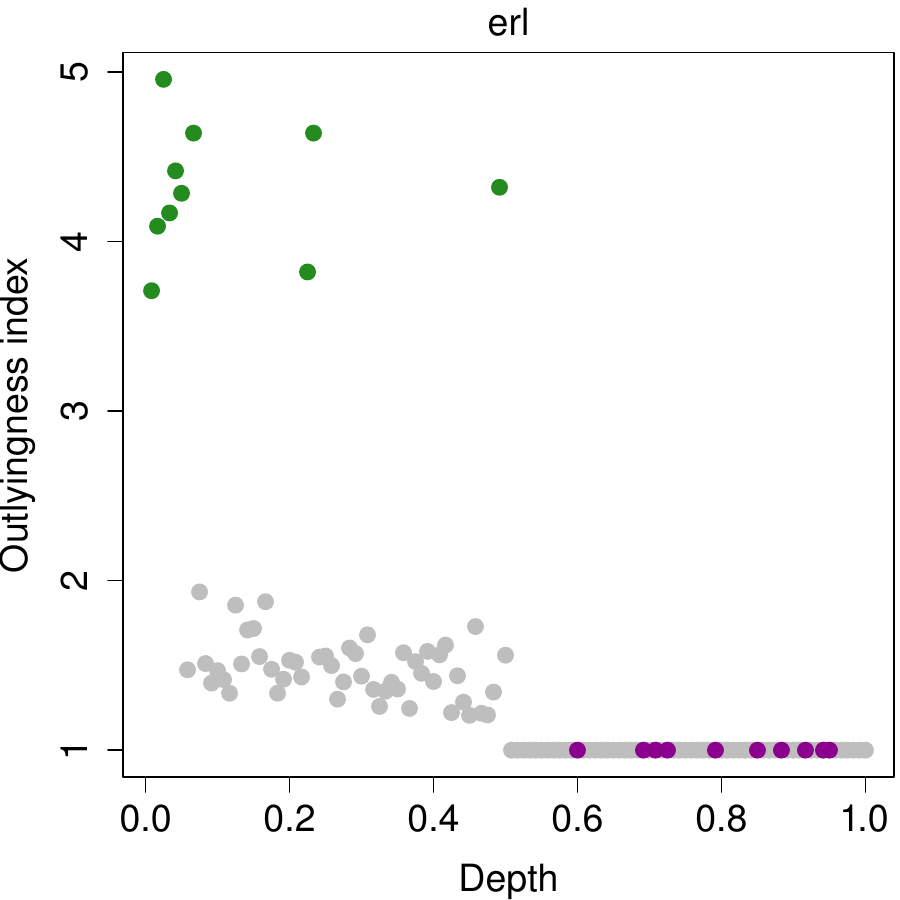} 
    \caption{Scatter plots of the depth ($x$-axis) against the outlyingness index of each function from our motivating example, using MBD (left) and \textsf{erl}  (right). Color coding corresponds to the colors used in Figure~\ref{MotExf}.}
    \label{MotExFB1}
\end{figure}

One expects that the functional boxplot based on an integrated depth will be sensitive particularly to the global outliers, whereas the functional boxplot based on an infimal depth will be sensitive to local outliers. Figure~\ref{MotExFB1} allows us to assess this. On the $x$-axis, we see the depth attached to the data. On the $y$-axis, we display the outlyingness index of the curves, defined as the maximum inflation factor of the $50\%$ central region that makes the datum fall outside the inflated central region. Outlyingness index $4$ corresponds to the usual constant $1.5\times\mathrm{IQR}$ from the one-dimensional boxplot. As we see in Figure~\ref{MotExFB1}, MBD is indeed able to detect the global outliers (the curves in magenta get very low MBD on the $x$-axis), but its boxplot ($y$-axis) does not flag these functions as outlying. The figure also shows that the \textsf{erl} index is sensitive to local outliers (green points with low depth values on the $x$-axis), and the same is true for its boxplot ($y$-axis). Thus, the functional boxplot with an infimal depth is capable of detecting the local outliers, but the functional boxplot with an integrated depth fails to detect even the global outliers that are identified as extreme by the integrated depth. It is unrealistic to expect a functional boxplot to detect outliers that are not identified as extreme by the used depth; thus, we do not anticipate this. We, however, see that the boxplot based on an integrated depth fails to flag even the outliers detected by its own depth, while the boxplot based on an infimal depth does not suffer from such problems. Our example illustrates two fundamental problems of MBD-based boxplots: 
\begin{description}
    \item[\textbf{Problem~1:}] \textbf{Lack of robustness.} Local outliers can seriously bloat the central region computed by an integrated depth. This happens because local outliers can obtain high integrated depth, and thus contribute to the shape of the central region. The subsequent expansion of the $50\%$ central region to obtain the whiskers band then makes the functional boxplot insensitive to any deviations in the data.  
    \item[\textbf{Problem~2:}] \textbf{Outlier masking.} Even though the global outliers have the lowest MBD of all sample functions (see the left-hand panel of Figure~\ref{MotExFB1}), they fail to be detected as outliers by the MBD boxplot. Even worse, they are completely contained in the $50\%$ MBD central region.
\end{description} 
The gist of our example shows that the integrated depth-based central regions do not satisfy a natural property one would expect from a boxplot based on bands:

\smallskip

(\textcolor{red}{\textsf{BC}}) \textbf{Band convexity:} A function with depth lower than all functions used to construct the central region $C$ must leave $C$ at least at one argument of the function.

\smallskip
    
We call this crucial property \emph{band convexity}; in other works, it is also called the \emph{intrinsic graphical interpretation} property \citep{Myllymaki_Mrkvicka2019}. 

\begin{figure}
    \centering
    \includegraphics[width=0.45\linewidth]{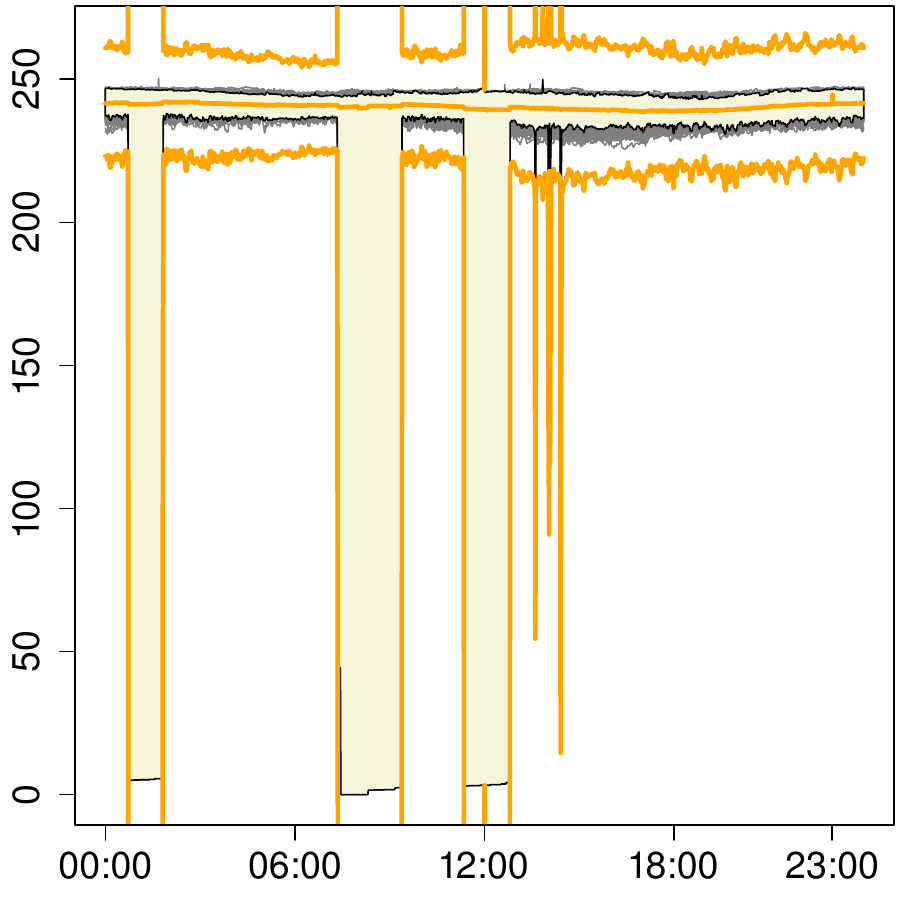} 
    \includegraphics[width=0.45\linewidth]{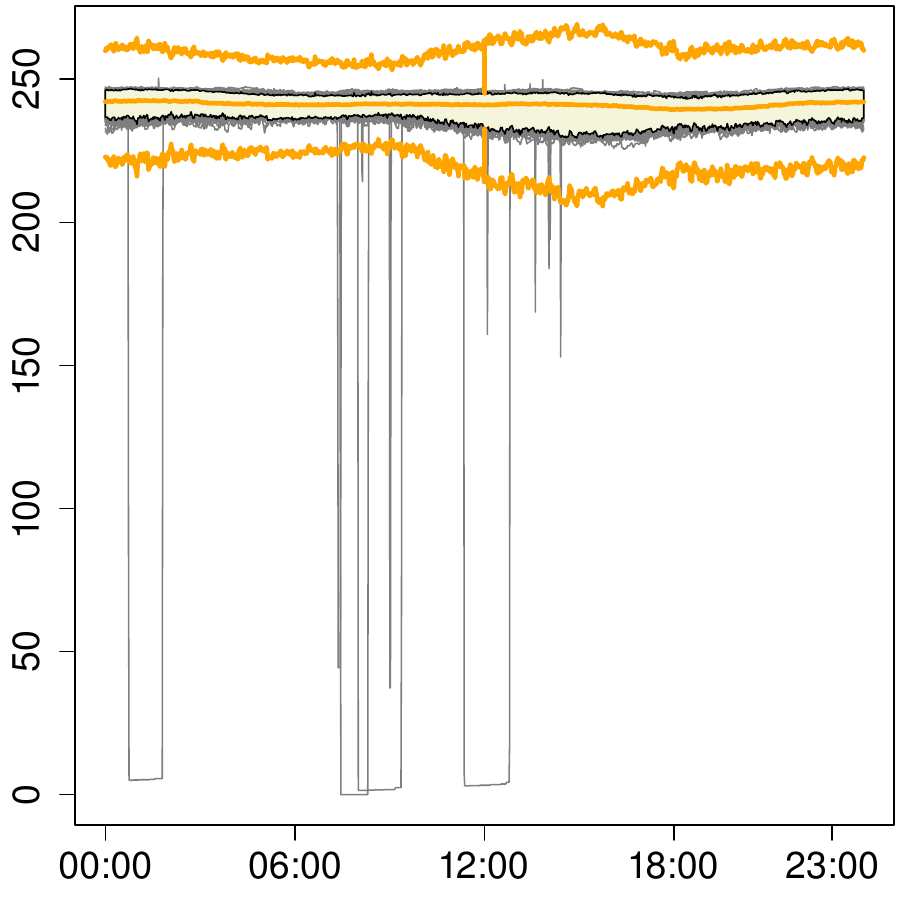} 
    \caption{Real data example: Household voltage circuit monitoring using an integrated depth-based boxplot (left) and an infimal depth-based boxplot (right). For a detailed description of the data see the Supplementary Material.}
    \label{figure: real data}
\end{figure}

Our motivating example is not unrealistic --- a real dataset exhibiting the same phenomenon can be found in Figure~\ref{figure: real data}. Both these examples illustrate our main misgivings concerning using integrated depths together with a functional boxplot. Several additional, more elaborate examples in this spirit can be found in the Supplementary Material. The rest of this paper contains theoretical proofs supporting and refining our arguments. First, after introducing the necessary background on bands and functional depths in Sections~\ref{section: band}--\ref{section: band convexity}, in Section~\ref{section: central region} we prove that the central region constructed from an integrated depth and a finite sample of functions is inconsistent as the number of functions increases; in fact, it converges to the whole function space. Second, in Section~\ref{section: infimal depth} we prove that the band convexity of a depth is satisfied exclusively for depths of infimal type. Third, we prove that infimal depth satisfies the rearrangement invariance, another natural desideratum for the construction of a functional boxplot. Finally, in Section~\ref{section: one-dimensional} we specify the properties of one-dimensional depths that are used to define the infimal depths, in relation to the properties of the resulting boxplot. A brief simulation study can be found in Section~\ref{section: simulations}. There, we further explore the statistical properties of various boxplots. In particular, we demonstrate 
that the population MBD central region depends heavily on the number of functions from which it is constructed as well as on the amount of smoothness of the functions. The same problem does not appear for boxplots based on infimal depths. Some extensions of our theory and concluding remarks are in Section~\ref{section: conclusions}. Technical details and additional supporting results are postponed to the Supplementary Material.

\section{Functional boxplots and functional depth}

Let $\left(\Omega, \mathcal A, \PP\right)$ be the probability space on which all random variables are defined. For a topological space $\mathcal S$, $\Prob[\mathcal S]$ stands for all Borel probability measures on $\mathcal S$, and $X \sim P \in \Prob[\mathcal S]$ denotes a random variable $X$ with distribution $P$. The one-dimensional Lebesgue measure on $\R$ is denoted by $\lambda$.

We consider data living in $\FS$, a Banach space of functions $\mathcal T \to \R$ for $\mathcal T \subset \R$ a compact interval; without loss of generality, we take $\mathcal T = [0,1]$. Only very mild properties are imposed on the function space $\FS$. We assume that constant functions belong to $\FS$, and that for any $x, y \in \FS$ we also have $x\,y$ defined by $t \mapsto x(t)\,y(t)$ as an element of $\FS$. We further assume that the norm $\left\Vert \cdot \right\Vert$ in $\FS$ is \emph{compatible with} the absolute value $\left\vert \cdot \right\vert$ in $\R$. This means that whenever $\left\vert u_n \right\vert \to \infty$ in $\R$, then for every sequence of functions $\left\{ x_n \right\}_{n=1}^\infty \subset \FS$ such that $x_n(t) = u_n$ for some $t \in [0,1]$ we can write $\left\Vert x_n \right\Vert \to \infty$. Finally, we require a weak \emph{approximation property} 
    \begin{enumerate}[label=\textcolor{red}{\textsf{(A)}}, ref=\textsf{(A)}]
    \item \label{A} For any $f \in \FS$, $t_0 \in [0,1]$, $y \in \R$, and $\varepsilon > 0$ there exists a function $g = g(f,t_0,y,\varepsilon) \in \FS$ such that $g(t_0) = y$, and $\lambda\left( \left\{ t \in [0,1] \colon g(t) \ne f(t) \right\} \right) < \varepsilon$.
    \end{enumerate}
This property is satisfied for all interesting function spaces. For instance, for $\FS$ the space of continuous functions,~\ref{A} holds true thanks to the classical Luzin theorem \citep[Theorem~7.5.2]{Dudley2002} and the Tietze-Urysohn extension theorem \citep[Theorem~2.6.4]{Dudley2002}. The compatibility of the supremum norm $\left\Vert x \right\Vert = \sup_{t \in [0,1]} \left\vert x(t) \right\vert$ used in the space of continuous functions is trivially guaranteed as $\left\vert x(t) \right\vert \leq \left\Vert x \right\Vert$ for each $t \in [0,1]$.

For a functional distribution $P \in \Prob[\FS]$ with $X \sim P$ we denote by $P_t \in \Prob[\R]$ the distribution of the random variable $X(t) \in \R$, for $t \in [0,1]$.

\subsection{Bands of functions} \label{section: band}

Band depths for functional data are inherently tied to the concept of bands in function spaces. A \emph{band} of a collection of functions $B$ is a set in $\FS$ of the form
    \begin{equation}    \label{eq: band}
    \begin{aligned}
    & \band{B} \\
    & = \Big\{ x \in \FS \colon \inf_{y \in B} y(t) \leq x(t) \leq \sup_{y \in B} y(t) \mbox{ for all }t \in [0,1]\Big\}.   
    \end{aligned}
    \end{equation}  
The delimiting functions $y \colon [0,1] \to \Rinf$ from $B$ do not necessarily come from $\FS$; they are allowed to take infinite values in $\Rinf = \R \cup \left\{ - \infty, \infty \right\}$. A trivial example of a band is the whole space $\FS$, if $B$ is composed of $y_1(t) = -\infty$ and $y_2(t) = \infty$ for all $t \in [0,1]$. In the common situation when $B$ is finite with elements $y_1, \dots, y_m$ we write also $\band{y_1, \dots, y_m}$ instead of $\band{\left\{y_1, \dots, y_m\right\}}$. The set 
    \[  
    \begin{aligned}
    & \left\{ u \in \R \colon \inf_{x \in \band{B}} x(t) \leq u \leq \sup_{x \in \band{B}} x(t) \right\} \\
    & = \left\{ x(t) \colon x \in \band{B} \right\} \subseteq \R  
    \end{aligned} \]
is called the \emph{slice} of the band $\band{B}$ at $t \in [0,1]$. Bands of functions are closely related to the envelopes of functional data \citep{Myllymaki_Mrkvicka2019}. 

\subsection{One-dimensional depths}

The theoretical background of depths defined on Euclidean spaces $\R^d$ is well studied \citep{Liu1990, Zuo_Serfling2000}. According to \citet{Zuo_Serfling2000}, a depth function $D$ in $\R^d$ should satisfy the following minimal criteria: 
    \begin{enumerate*}[label=(\arabic*{})]
    \item Affine invariance; 
    \item Maximality at the center; 
    \item Monotonicity w.r.t. the deepest point; and
    \item Vanishing at infinity.
    \end{enumerate*}
Restrict now the argumentation to the one-dimensional case $d=1$. For a one-dimensional distribution $Q \in \Prob[\R]$ we will understand the median of $Q$ in the broad sense as any $u \in \R$ such that $\lim_{v \to u-} F(v) \leq 1/2 \leq F(u)$, where $F$ is the distribution function of $Q$; we write $u \in \median{Q}$. A one-dimensional depth $\Du \colon \R \times \Prob[\R] \to [0,1]$ is expected to verify for each $Q \equiv Q_Y \in \Prob[\R]$ and $Y \sim Q$
    \begin{enumerate}[label=\textcolor{red}{\textsf{(D$_\arabic*$)}}, ref=\textsf{(D$_\arabic*$)}]
        \item \label{D1} \textbf{Affine invariance:} For any $a \in \R$, $a \ne 0$ and any $b \in \R$ we have
            \[  \Du(u; Q_Y) = \Du(a\,u + b; Q_{a\,Y + b})   \quad \mbox{for all }u \in \R,  \]
        where $Q_{a\,Y + b} \in \Prob[\R]$ stands for the distribution of $a\,Y + b$.
        \item \label{D2} \textbf{Maximality at the median:} If for $m \in \R$ we have
            \begin{equation}    \label{eq: univariate maximality}
            \Du(m; Q) = \sup_{u \in \R} \Du(u; Q) > 0,   
            \end{equation}
        then $m \in \median{Q} \in \R$, i.e. $m$ is a median of $Q$.
        \item \label{D3} \textbf{Monotonicity w.r.t. the median:} For any $m \in \R$ that satisfies~\eqref{eq: univariate maximality}, the function $u \mapsto \Du(u; Q)$ is non-decreasing on $(-\infty, m]$ and non-increasing on $[m, \infty)$.
        \item \label{D4} \textbf{Vanishing at infinity:} 
        \[ \lim_{u \to \infty} \Du(u; Q) = \lim_{u \to -\infty} \Du(u; Q) = 0.\]
    \end{enumerate}
Classical examples of one-dimensional depths are the \emph{halfspace} (or \emph{Tukey}) depth in $\R$ \citep{Tukey1975, Donoho_Gasko1992} given by
    \begin{equation}    \label{eq: halfspace depth}  
    \DuH(u; Q) = \min\left\{ Q((-\infty, u]), Q([u,\infty)) \right\}  
    \end{equation}
for $u \in \R$, or the one-dimensional \emph{simplicial depth} \citep{Liu1990} defined as
    \[  \DuS(u; Q) = \PP\left( u \in \left[ \min\left\{ Y_1, Y_2 \right\}, \max\left\{ Y_1, Y_2\right\} \right] \right)  \]
for $u \in \R$, where $Y_1, Y_2 \sim Q$ are independent. The simplicial depth can be equivalently written in terms of the distribution function $F$ of $Q$ as
    \begin{equation}    \label{eq: simplicial depth}
    \begin{aligned}
    \DuS(u; Q) & = 1 - Q((-\infty,u))^2 - Q((u,\infty))^2 \\
    & = 2\,F(u)(1-F(u)) + Q(\left\{u\right\})^2
    \end{aligned}
    \end{equation}
for $u \in \R$. The halfspace depth~\eqref{eq: halfspace depth} satisfies all conditions~\ref{D1}--\ref{D4} for any $Q \in \Prob[\R]$. The simplicial depth~\eqref{eq: simplicial depth} verifies~\ref{D1} and~\ref{D4} for any $Q \in \Prob[\R]$, but~\ref{D2} and~\ref{D3} only if $Q \in \Prob[\R]$ does not contain atoms; see the Supplementary Material.

\subsection{Functional depths}  \label{section: functional depth}

A general \emph{functional depth} is a mapping 
    \[  D \colon \FS \times \Prob[\FS] \to [0,1]    \] 
that satisfies (some of) the desirable properties listed in \cite{Gijbels_Nagy2017}. These properties directly translate~\ref{D1}--\ref{D4} required from one-dimensional, or more generally, $d$-variate depth from \citet{Zuo_Serfling2000}. The following versions of these conditions will be important in the sequel. For every $P \equiv P_X \in \Prob[\FS]$ and $X \sim P$ we want $D$ to satisfy:
    \begin{enumerate}[label=\textcolor{red}{\textsf{(FD$_\arabic*$)}}, ref=\textsf{(FD$_\arabic*$)}]
        \item \label{FD1} \textbf{Function-affine invariance:} For any $a, b \in \FS$, $a(t) \ne 0$ for all $t \in [0,1]$ we have
            \[  D(x; P_X) = D(a\,x + b; P_{a\,X + b})   \quad \mbox{for all }x \in \FS,  \]
        where $P_{a\,X + b} \in \Prob[\FS]$ stands for the distribution of $a\,X + b \colon t \mapsto a(t)\,X(t) + b(t)$.
        \item \label{FD2} \textbf{Maximality at the coordinate-wise median:} If for some $m \in \FS$ 
            \begin{equation}    \label{eq: function maximality}
            D(m; P) = \sup_{x \in \FS} D(x; P) > 0   
            \end{equation}
        then $m(t) \in \median{P_t}$ for all $t \in [0,1]$, i.e. $m$ must be a coordinate-wise median of $P$.
        \item \label{FD3} \textbf{Band monotonicity w.r.t. the median:} For any $m \in \FS$ that satisfies~\eqref{eq: function maximality}, any $x \in \FS$ and any $y \in \FS \cap \band{m,x}$ we have $D(x; P) \leq D(y; P)$.
        \item \label{FD4} \textbf{Vanishing at infinity:} $\lim_{\left\Vert x \right\Vert \to \infty} D(x; P) = 0$.
        \end{enumerate}

The \emph{band depth} \citep{Lopez_Romo2009} of a function $x \in \FS$ w.r.t. $P \in \Prob[\FS]$ is given by
    \begin{equation}    \label{eq: band depth}
    \BD(x;P) = \PP\left( x \in \band{X_1, X_2} \right), 
    \end{equation}      
where $X_1, X_2 \sim P$ are independent random functions.\footnote{This definition corresponds to the common choice of $J=2$ random functions $X_1, X_2$ in \citet{Lopez_Romo2009}. Analogous band depths for $J > 2$ can also be considered, with only minor changes in what follows.} The band depth presents a functional analog of the simplicial depth from $\R^d$, see~\eqref{eq: simplicial depth}. A depth even more popular in the construction of functional boxplots is the \emph{modified band depth} \citep[MBD,][]{Lopez_Romo2009}, defined for $x \in \FS$ w.r.t. $P \in \Prob[\FS]$ as
    \begin{equation}    \label{eq: MBD}  
    \begin{aligned}
    \MBD(x;P) = \mathsf{E}\,  & \left[\lambda\left( \left\{ t \in [0,1] \colon \min\{ X_1(t), X_2(t) \} \right. \right. \right. \\
    & \left. \left. \left. \leq x(t) \leq \max\{ X_1(t), X_2(t) \}  \right\} \right)\right]. 
    \end{aligned}
    \end{equation}
Again, $X_1, X_2 \sim P$ are independent. The integrand in~\eqref{eq: MBD} is interpreted as the proportion of time points $t \in [0,1]$ such that $x$ is contained inside the (random) band $\band{X_1, X_2}$. MBD is a particular case of a wider collection of functional depths of \emph{integrated type} (or \emph{integrated depths}) \citep{Fraiman_Muniz2001, Cuevas_Fraiman2009, Nagy_etal2016}. An integrated depth is always defined with a one-dimensional depth as its building block. First, a one-dimensional depth $\Du \colon \R \times \Prob[\R] \to [0,1]$ is given. Then a \emph{integrated depth} based on $\Du$ is defined for $x \in \FS$ and $P \in \Prob[\FS]$ as 
    \begin{equation}  D(x; P) = \int_0^1 \Du(x(t); P_t) \dd t. \label{eq: integrated depth}    
    \end{equation}
It is simple to observe \citep[Section~A.3]{Nagy_etal2016} that MBD from~\eqref{eq: MBD} can be written as the integrated depth~\eqref{eq: integrated depth} based on the one-dimensional simplicial depth~\eqref{eq: simplicial depth}.

Another important family of functional depths is that of \emph{infimal depths} given by
    \begin{equation}    \label{eq: infimal depth}
    D(x; P) = \inf_{t \in [0,1]} \Du(x(t); P_t).
    \end{equation}
Infimal depths were first considered in \citet{Mosler2013} and \citet{Mosler_Polyakova2016}. Under different names, they have been used and amended in \citet{Narisetty_Nair2016, Myllymaki_etal2017} and \citet{Myllymaki_Mrkvicka2019}.

A major advantage of integrated and infimal depths is their computational simplicity --- as the common one-dimensional depths $\Du$ are computable in linear time $O(n)$ for $n$ the sample size, both functional depths~\eqref{eq: integrated depth} and~\eqref{eq: infimal depth} are of complexity $O(n)$ too. This should be compared with the band depth~\eqref{eq: band depth}, whose naive computation requires $O(n^2)$ operations.
    



%
%
%

\subsection{Band convexity of functional depths}    \label{section: band convexity}

Given a dataset $\calX$ of independent and identically distributed functions $X_1, \dots, X_n$ sampled from $P \in \Prob[\FS]$, a functional depth $D$ is used to rank the observations in $\calX$ from the most central toward the peripheral ones. $\calX$ are used to construct the empirical distribution $P_n \in \Prob[\FS]$ that gives weight $1/n$ to each $X_i$, and to assess the (functional) depth of each observation
    \[  D_i = D(X_i; P_n), \quad \mbox{for }i = 1,\dots,n. \]
The higher $D_i$ is, the more centrally positioned is $X_i$ inside the dataset. We obtain an ordering $X_{1:n}, \dots, X_{n:n}$ of $\calX$ according to the decreasing depth value. The function $X_{1:n}$ is defined to be $X_i$ such that $D_i$ is maximal among $D_1, \dots, D_n$; this function is termed the \emph{median function} of $\calX$. $X_{2:n}$ is the $X_i$ whose depth $D_i$ is the second largest, etc. Cases of ties are broken as in $\R$ --- if a common value $D_i$ is shared by $m$ functions from $\calX$, then all these functions are assigned the average rank.  

The key step in constructing the functional boxplot in \citet{Sun_Genton2011} is defining the central region. This is done using a functional depth. In general, the \emph{$\tau$-sample central region} of a dataset $\calX$ with $\tau \in (0,1)$ is given as
    \begin{equation}    \label{eq: tau sample central region}
    \band{X_{1:n}, \dots, X_{\ceil{\tau\,n}:n}},    
    \end{equation}     
that is the band of the $\ceil{\tau\,n}$ deepest sample functions. Even though any functional depth can be used in a functional boxplot, the central region in~\eqref{eq: tau sample central region} is always a band. For the specific construction of a functional boxplot, the $\tau$-central region is used with $\tau = 1/2$. That is, one half of the deepest sample functions are taken in the construction of the central region.

When considering the $\tau$-central region in~\eqref{eq: tau sample central region}, it is implicitly assumed that any function $x$ inside this band gets depth $D(x; P_n)$ not smaller than $D(X_{\ceil{\tau\,n}}; P_n)$, i.e., it is at least as deep inside $P_n$ as any of the functions $X_{1:n}, \dots, X_{\ceil{\tau\,n}:n}$. Not all functional depths, however, possess this property. Recall our motivating example from Section~\ref{section: introduction}, where all the global outliers, which have the smallest MBD, are contained inside the $50\%$ MBD central region.

In analogy with the common convexity requirement for depths in $\R^d$ \citep{Serfling2006}, we require the \emph{band convexity} from functional depths:
    \begin{enumerate}[label=\textcolor{red}{\textsf{(BC)}}, ref=\textsf{(BC)}]
    \item \label{BC} For any $P \in \Prob[\FS]$ and any $\alpha \in [0,1]$, the central region 
        \begin{equation}    \label{eq: central region}
        D_\alpha(P) = \left\{ x \in \FS \colon D(x;P) \geq \alpha \right\}    
        \end{equation}
    is a band, i.e. $D_\alpha(P) = \band{D_\alpha(P)}$.
    \end{enumerate}
This condition is a formal translation of~\ref{BC} imposed in Section~\ref{section: introduction}. 

Observe that the indexing in the sample case~\eqref{eq: tau sample central region} and the population case~\eqref{eq: central region} is different --- the constant $\alpha$ in~\eqref{eq: central region} corresponds to the $(1-\tau)$-quantile of the depth distribution $D(X; P)$, where $X \sim P$ is random. We denote the latter value by $\alpha(\tau) \in [0,1]$; $\alpha \colon (0,1) \to [0,1]$ is then a non-increasing function of $\tau \in (0,1)$. The quantity $\alpha(\tau)$ is designed so that approximately with probability $\tau$, a random function $X \sim P$ will have depth $D(X; P)$ greater than $\alpha(\tau)$.

Our motivating example showed that integrated depths do not satisfy~\ref{BC}. An example showing that also the standard band depth~\eqref{eq: band depth} violates~\ref{BC} is presented in the Supplementary Material.

\subsection{The central region of an integrated depth}  \label{section: central region}

A fundamental problem with the functional boxplot based on an integrated depth concerns its population version. We have a dataset $\calX$ of $X_1, \dots, X_n$ sampled from $P \in \Prob[\FS]$. We ask about the analog of the boxplot applicable directly to $P$, instead of just the sample boxplot. A natural population counterpart to the central region \eqref{eq: tau sample central region} with $\tau = 1/2$ is the band of all functions $x \in \FS$ whose depth is at least $\alpha(1/2)$, the median depth of a random function $X \sim P$. The population version of \eqref{eq: tau sample central region} is then the band
    \begin{equation*}    
    \band{D_{\alpha(1/2)}(P)} = \band{\left\{ x \in \FS \colon D(x; P) \geq \alpha(1/2) \right\}},    
    \end{equation*}   
where we used the notation from~\eqref{eq: central region}. We have the following observation to make.

\begin{theorem} \label{theorem: infinite band}
Let $D$ be any integrated functional depth from~\eqref{eq: integrated depth} based on $D_1$ that obeys~\ref{D2} and~\ref{D3}. Then for any $P \in \Prob[\FS]$, the population $\alpha$-central region \eqref{eq: central region} for $\alpha < \sup_{x \in \FS} D(x; P)$ is the whole function space $\FS$.
\end{theorem}

\begin{proof}
Directly from the definition of the integrated depth $D$ from~\eqref{eq: integrated depth} and conditions~\ref{D2} and~\ref{D3} imposed on $\Du$, we see that the maximum integrated depth is attained at (any) coordinate-wise median function $m \colon t \mapsto m(t) \in \median{P_t}$. Without loss of generality, assume that such $m$ is unique and $m \in \FS$.\footnote{It may happen that $m \notin \FS$. 
In that case, one defines $\alpha^* = \sup_{x \in \FS}D(x; P)$ and proceeds analogously.} The maximum integrated depth value of any function is $\alpha^* = \int_0^1 \Du(m(t); P_t) \dd t > 0$. 

Suppose now that $\alpha \leq \alpha^*$. Take $\varepsilon \in (0, \alpha^* - \alpha)$, $y \in \R$, and $t \in [0,1]$ be arbitrary. By property~\ref{A}, we know that there exists a function $x_t \in \FS$ such that $x_t(t) = y$, and at the same time $x_t$ is in ``most'' of the domain $[0,1]$ identical with $m$. The definition of the integrated depth~\eqref{eq: integrated depth} then guarantees that such a function $x_t$ can be found so that $D(x_t; P) \geq \alpha^* - \varepsilon$. In particular, $x_t \in D_{\alpha}(P)$. The function $x_t$ then necessarily contributes to the set $\band{D_\alpha(P)}$, yet as $x_t(t)$ can be taken to be any value $y \in \R$, we get that slice of the band $\band{D_\alpha(P)}$ at $t \in [0,1]$ is, in fact, the whole space $\R$. Because this construction is possible to be performed for any $t \in [0,1]$, the resulting band $\band{D_\alpha(P)}$ must span the whole space $\FS$.
\end{proof}

Theorem~\ref{theorem: infinite band} holds true for the functional boxplot defined by MBD from~\eqref{eq: MBD}. The population version of that boxplot is, therefore, not well-defined. In fact, it is unclear what, in this situation, the sample boxplot based on the central region~\eqref{eq: tau sample central region} with $\tau = 1/2$ even estimates. Our motivating example demonstrated this: the local outliers belonged to the MBD $50\%$ central region. The same phenomenon is also attested in the Supplementary Material, where we run a small simulation exercise to account for the stochastic nature of the setup. In that example, the width of the $50\%$ central region based on an integrated and infimal depth, together with the average number of detected outliers, is reported. As expected, the mean width grows with the sample size $n$ for integrated depths, but remains stable when infimal depths are used.

\subsection{A necessary condition for band convexity}   \label{section: infimal depth}

Our main result states that the only functional depths that satisfy the crucial band convexity criterion~\ref{BC} are depths of infimal type. To prove it in the following theorem, we will argue that a function $g \colon \R \to [0,1]$, that constitutes a building block for the infimal depth, can be considered as a one-dimensional depth function that satisfies conditions~\ref{D1}--\ref{D4}. Since a depth $D_1$ is a function of both $u \in \R$ and $Q \in \Prob[\R]$ but $g$ is a function of only $u \in \R$, we must first explain what this means. Out of the desiderata for $D_1$, only~\ref{D1} explicitly depends on the argument of the distribution $Q \in \Prob[\R]$; the other conditions~\ref{D2}--\ref{D4} take $Q$ fixed and given. Thus, by saying that $g$ verifies~\ref{D2}--\ref{D4}, we mean that $g$ satisfies these conditions as a function of $u \in \R$. We say that $g$ satisfies~\ref{D1}, if for all $a \ne 0$ and $b \in \R$ given, $g$ can be defined in two setups: (i) $g_{1,0}$ for $Q_Y \in \Prob[\R]$ corresponding to $Y$, and (ii) $g_{a,b}$ for $Q_{a\,Y + b} \in \Prob[\R]$ corresponding to $a\,Y+b$. We say that $g$ verifies~\ref{D1} if $g_{1,0}(u) = g_{a,b}(a\,u+b)$ for all $u \in \R$, $a \ne 0$, and $b \in \R$. 

\begin{theorem} \label{theorem: infimal}
Suppose that a map $D \colon \FS \times \Prob[\FS] \to [0,1]$ verifies~\ref{BC}. Then for each $\alpha \in [0,1]$ we can write $D_\alpha(P) = \band{\ell_\alpha, u_\alpha}$ for some functions $\ell_\alpha, u_\alpha \colon [0,1] \to \Rinf$, $\ell_\alpha \leq u_\alpha$, and $D$ can be written in the form 
    \begin{equation}    \label{eq: g infimal depth}
    D(x; P) = \inf_{t \in [0,1]} g_t(x(t)),    
    \end{equation} 
where
    \begin{equation}    \label{eq: gt}
     g_t(u) = \sup\left\{ \alpha \in [0,1] \colon u \in \left[ \ell_\alpha(t), u_\alpha(t) \right] \right\}      
    \end{equation}  
for $u \in \R$ is a collection of one-dimensional functions $g_t \colon \R \to [0,1]$, $t \in [0,1]$. In addition, the following holds true.
    \begin{enumerate}[label=(\roman*), ref=(\roman*)]
    \item \label{theorem: infimal i} If $D$ satisfies condition~\ref{FD1}, then each $g_t$ verifies~\ref{D1};
    \item \label{theorem: infimal ii} If $D$ satisfies condition~\ref{FD2}, then each $g_t$ verifies~\ref{D2};
    \item \label{theorem: infimal iii} $D$ satisfies~\ref{FD3}. 
    \item \label{theorem: infimal iv} Each function $g_t$ is upper semi-continuous and satisfies~\ref{D3}.
    \item \label{theorem: infimal v} $D$ satisfies~\ref{FD4} if and only if for each $\alpha > 0$ and $P \in \Prob[\FS]$ is $D_\alpha(P)$ bounded. 
    \item \label{theorem: infimal vi} Suppose that $D$ satisfies~\ref{FD4} and $\FS$ verifies
        \begin{enumerate}[label=\textcolor{red}{\textsf{(A')}}, ref=\textsf{(A')}]
        \item \label{A'} For any $f \in \FS$, $t_0 \in [0,1]$, $y \in \R$, we have that a function $g$ given by $g(t_0) = y$ and $g(t) = f(t)$ for $t \ne t_0$ belongs to $\FS$.
        \end{enumerate}  
    Then each function $g_t$ verifies~\ref{D4}. 
    \end{enumerate}
In particular, if $D$ satisfies all~\ref{FD1}--\ref{FD4} and~\ref{A'} is true for $\FS$, then $g_t$ defined in~\eqref{eq: gt} is a depth in $\R$ that satisfies all~\ref{D1}--\ref{D4}. In that case, $D$ must be an infimal depth for functional data.
\end{theorem}

The opposite implication of Theorem~\ref{theorem: infimal} is also true, under mild conditions --- a functional depth $D$ verifies~\ref{BC} if and only if it is of infimal type. We prove this in Theorem~\ref{theorem: properties of infimal depth} below.

\begin{proof}
\textbf{Formula~\eqref{eq: g infimal depth}.} We adapt the argumentation from \citet[Proposition~5]{Dyckerhoff2004} to the function case, and use the structure of functional bands. The map $D$ can equivalently be written in the form
    $D(x; P) = \sup\left\{ \alpha \in [0,1] \colon x \in D_\alpha(P) \right\}$,
where $D_\alpha(P)$ are the central regions of $D$ given in \eqref{eq: central region}. Since each $D_\alpha(P)$ is a band, it is uniquely determined by its lower and upper envelope functions, respectively, given by
    \begin{equation}    \label{eq: l and u}
    \begin{aligned}
    \ell_\alpha(t) & = \inf\left\{ x(t) \colon x \in D_\alpha(P) \right\}, \quad \mbox{ and } \\
    u_\alpha(t) & = \sup\left\{ x(t) \colon x \in D_\alpha(P) \right\}.
    \end{aligned}
    \end{equation}
Thanks to the condition~\ref{BC}, we see in particular that 
    \[  D_\alpha(P) = \left\{ x \in \FS \colon \ell_\alpha(t) \leq x(t) \leq u_\alpha(t) \mbox{ for all }t \in [0,1] \right\}   \]
for each $\alpha \in [0,1]$. We thus have that $x \in D_\alpha(P)$ if and only if $x(t) \in \left[ \ell_\alpha(t), u_\alpha(t) \right]$ for all $t \in [0,1]$ and $x \in \FS$. We can now define a function $g_t \colon \R \to [0,1]$ by~\eqref{eq: gt}. This function depends only on the slice of $D_\alpha(P)$ at $t \in [0,1]$. The map $g_t$ is a function of only $x(t)$. Using $g_t$ we have $x \in D_\alpha(P)$ if and only if $g_t(x(t)) \geq \alpha$ for each $t \in [0,1]$, or equivalently if $\inf_{t \in [0,1]} g_t(x(t)) \geq \alpha$. We obtain that for each $\alpha \in [0,1]$ we have $D(x; P) \geq \alpha$ if and only if $\inf_{t \in [0,1]} g_t(x(t)) \geq \alpha$, which means that necessarily $D(x; P) = \inf_{t \in [0,1]} g_t(x(t))$, as was to be shown in the first claim of the theorem. 

\mypar{Part~\ref{theorem: infimal i}} Suppose now that~\ref{FD1} is true for $D$. In terms of the level sets $D_\alpha(P)$, that condition can be written as function-affine equivariance of the sets $\left\{ D_\alpha(P) \right\}_{\alpha \in [0,1]}$. That equivariance is defined by
    \begin{equation}    \label{eq: function-affine equiv} 
    D_\alpha(Q) = \left\{ a\,x + b \in \FS \colon x \in D_\alpha(P) \right\}  
    \end{equation}
being true for all $a, b \in \FS$, $a(t) \ne 0$ for all $t \in [0,1]$, where $Q \sim \Prob[\FS]$ stands for the distribution of the function-affine image $a\,X + b \in \FS$ of $X \sim P$. Now, since~\eqref{eq: function-affine equiv} can also be written for each slice at $t \in [0,1]$ separately, an analogous property must be true for each $t \in [0,1]$ fixed. The latter can be expressed as the affine invariance of $g_t$ as in~\ref{D1}.

\mypar{Part~\ref{theorem: infimal ii}} Next, we want to show that~\ref{FD2} for $D$ implies~\ref{D2} for each $g_t$. Suppose that $\alpha^* = \sup_{x \in \FS} D(x; P) = D(m; P)$ as in~\ref{FD2}. Then 
    \[ \alpha^* = D(m; P) = \inf_{t \in [0,1]} g_t(m(t)) \]
by the first claim of this theorem, and necessarily for any $x \in \FS$ with $D(x; P) \geq \alpha^*$ we have $g_t(x(t)) \geq \alpha^*$ for each $t \in [0,1]$. By~\ref{FD2} we know that any such $x$ must be a coordinate-wise median function of $P$, i.e. $x(t) \in \median{P_t}$ for each $t \in [0,1]$. Together, we obtain that for each $t \in [0,1]$, $g_t(x(t)) \geq \alpha^*$ implies that $x(t) \in \median{P_t}$, meaning that any maximizer of $g_t(\cdot)$ must necessarily be a median of $\median{P_t}$, as we wanted to show. 

\mypar{Part~\ref{theorem: infimal iii}} It is enough to realize that for any $m , x$ and $y$ that obey the conditions in~\ref{FD3} we have that for $\alpha = D(x; P)$ we can write $\alpha \leq D(m; P)$. That gives $x, m \in D_\alpha(P)$, and because of the band convexity~\ref{BC} of $D$ also $y \in D_\alpha(P)$.

\mypar{Part~\ref{theorem: infimal iv}} For each $t \in [0,1]$, the upper level sets of the one-dimensional functions $u \mapsto g_t(u)$ at level $\alpha \in [0,1]$ are the intervals $\left[ \ell_\alpha(t), u_\alpha(t) \right] \subseteq \Rinf$ from~\eqref{eq: l and u}. Each such interval is closed, which means that the function $g_t$ must be upper semi-continuous for each $t \in [0,1]$. Furthermore, since we assume that each such upper level set is an interval, the function $g_t$ must also be unimodal in the sense that~\ref{D3} holds.

\mypar{Part~\ref{theorem: infimal v}} Take a sequence $\left\{ x_n \right\}_{n=1}^\infty \subset \FS$ such that $\left\Vert x_n \right\Vert \to \infty$. Fix a set $D_\alpha(P)$ with $\alpha > 0$. By definition, $D_\alpha(P)$ is bounded in $\FS$ if and only if for any sequence $\left\{ x_n \right\}_{n=1}^\infty$ as above we have that $x_n \notin D_\alpha(P)$ for all $n$ large enough. Because we assume that this is true for all $\alpha > 0$, we see that $\left\Vert x_n \right\Vert \to \infty$ implies $D(x_n; P) \to 0$ as $n \to \infty$. The opposite implication is analogous.

\mypar{Part~\ref{theorem: infimal vi}} Suppose that~\ref{A'} and~\ref{FD4} are true, $t_0 \in [0,1]$ is given, and $\left\{ u_n \right\}_{n=1}^\infty \subset \R$ is a sequence such that $\left\vert u_n \right\vert \to \infty$ as $n \to \infty$. Take $x \in \FS$ such that $D(x; P) > 0$, and define $\left\{x_n\right\}_{n=1}^\infty \subset \FS$ so that $x_n(t_0) = u_n$, but $x_n(t) = x(t)$ for all $t \ne t_0$. These functions belong to $\FS$ because of~\ref{A'}. Since we assume that the norm $\left\Vert \cdot \right\Vert$ in $\FS$ is compatible with $\left\vert \cdot \right\vert$ in $\R$, we have that $\left\vert u_n \right\vert \to \infty$ implies $\left\Vert x_n \right\Vert \to \infty$. Thus, by the boundedness of all the bands $D_\alpha(P)$ (in the norm $\left\Vert \cdot \right\Vert$) from part~\ref{theorem: infimal v} of this proof, necessarily $D(x_n; P) \to 0$ as $n \to \infty$. On the other hand, because $\inf_{t \in [0,1], t \ne t_0} g_t(x_n(t)) \geq \inf_{t \in [0,1]} g_t(x(t)) = D(x; P) > 0$ we also know that 
    \[  
    \begin{aligned}
    D(x_n; P) & = \inf_{t \in [0,1]} g_t(x_n(t)) \\
    & = \min\left\{ g_{t_0}(x_n(t_0)), \inf_{t \in [0,1], t \ne t_0} g_t(x_n(t)) \right\}, 
    \end{aligned}  
    \]
which gives that necessarily $g_{t_0}(x_n(t_0)) = g_{t_0}(u_n) \to 0$ as $n \to \infty$, as we wanted to show.
\end{proof}

The additional assumption~\ref{A'} in part~\ref{theorem: infimal vi} of Theorem~\ref{theorem: infimal} is stronger than the approximation property~\ref{A}. Condition~\ref{A'} is not valid in the space of continuous functions; for the purposes of our theory,  this is only a minor technical nuisance. For the last part of Theorem~\ref{theorem: infimal}, condition~\ref{A'} is, however, needed, as we see in the following example.

\begin{example} \label{example: wiener}
Take $\FS$ the space of continuous functions with the supremum norm, $P \in \Prob[\FS]$ the distribution of a Wiener process on $[0,1]$. Consider two infimal depths: first, take $D$ the infimal depth~\eqref{eq: g infimal depth} based on $g_t$ the one-dimensional halfspace depth~\eqref{eq: halfspace depth}. It will be shown below in Theorem~\ref{theorem: properties of infimal depth} that $D$ satisfies all properties~\ref{FD1}--\ref{FD4}. For the second depth, we define $g_0(u; Q) = 1$ for all $u \in \R$ and $Q \in \Prob[\R]$, and leave $g_t \equiv \DuH$ the halfspace depth for all $t \ne 0$. It is not difficult to see that the resulting infimal depth of type~\eqref{eq: g infimal depth} with this choice of $g_t$ coincides with the first infimal depth based on the halfspace depth~\eqref{eq: halfspace depth} for all $t \in [0,1]$. This holds true due to the upper semi-continuity of the function $t \mapsto \DuH(x(t); P_t)$ for each $x \in \FS$, see \citet[Theorem~4.14]{Nagy_etal2016}. Thus, even though $D$ obeys~\ref{FD4}, it is not necessary that $g_0$ must satisfy~\ref{D4}.
\end{example}

Theorem~\ref{theorem: infimal} shows that a sensible functional boxplot can be defined only if the underlying functional depth is of infimal type. An infimal depth~\eqref{eq: g infimal depth} is, on the other hand, uniquely determined by the collection of the one-dimensional functions $g_t$ on $\R$, $t \in [0,1]$. While in principle, it is possible to combine different one-dimensional depths $g_t$ in~\eqref{eq: g infimal depth} as we did in Example~\ref{example: wiener}, it is certainly more natural to consider a single depth $\Du = g_t$ for each $t \in [0,1]$ in the definition of $D$. This observation can be formalized using another property that is sometimes deemed desirable for functional depth, the \emph{rearrangement invariance} introduced in \citet[Section~4]{Mosler_Polyakova2016}.
    \begin{enumerate}[label=\textcolor{red}{\textsf{(RI)}}, ref=\textsf{(RI)}]
    \item \label{RI} For any surjective map $\varphi \colon [0,1] \to [0,1]$ such that for each $x \in \FS$ we have for $y = x \circ \varphi \colon t \mapsto x(\varphi(t))$ that $y \in \FS$, it holds true that 
        \[  D(x; P_X) = D(x \circ \varphi; P_{X \circ \varphi}) \quad \mbox{for all }x \in \FS.    \]
    Here, $P_{X \circ \varphi} \in \Prob[\FS]$ stands for the distribution of $X \circ \varphi \colon t \mapsto X(\varphi(t))$, and $X \sim P_X$.
    \end{enumerate}
Observe that the invariance in~\ref{RI} is not required only w.r.t. (measure-preserving) bijective maps $\varphi$, as considered in \citet{Mosler_Polyakova2016}, but more broadly.

The idea of rearrangement invariance of a functional depth couples tightly with the fact that all marginal depths $g_t$ defining the infimal depth $D$ are the same. 

\begin{theorem} \label{theorem: RI}
If we define a functional depth using the map~\eqref{eq: g infimal depth} with $\Du = g_t$ for each $t \in [0,1]$, then $D$ satisfies~\ref{RI}.
\end{theorem}

\begin{proof}
Naturally, the infimum in~\eqref{eq: g infimal depth} does not depend on the parametrization of the domain $[0,1]$; we have
    \[  
    \begin{aligned}
    D(x;  & P_X) = \inf_{t \in [0,1]} \Du(x(t); P_t) \\
    & = \inf_{t \in [0,1]} \Du(x(\varphi(t)); P_{\varphi(t)}) = D(x \circ \varphi; P_{X \circ \varphi}),    
    \end{aligned}  \]
where we used only the definition of the infimum and the surjectivity of $\varphi$.
\end{proof}

Also, the other implication in Theorem~\ref{theorem: RI} is true, under natural additional assumptions. 
Formulating such a result rigorously is, however, rather tedious as one has to cope with technical issues, such as the existence of surjective maps $\varphi \colon [0,1] \to [0,1]$ such that $x \circ \varphi \notin \FS$ (say, for $\FS$ continuous functions, if $\varphi$ fails to be continuous). This can be overcome at the cost of additional technical assumptions imposed on $\FS$; we refrain from stating this result explicitly.

\subsection{One-dimensional depth and its properties} \label{section: one-dimensional}

In the common situation when all marginal depths $g_t$ in~\eqref{eq: g infimal depth} are the same $\Du$, the properties of a well-defined functional boxplot depend only on the one-dimensional depth $\Du$. In what follows, we thus focus on one-dimensional depth functions and explore their properties w.r.t. the resulting functional boxplot. 

In the following result, we characterize all one-di\-men\-sio\-nal depth functions that verify all conditions~\ref{D1}--\ref{D4}. For that reason, it will be instrumental to partition the space of all probability distributions $\Prob[\R]$ into what we call \emph{affine classes}. A collection $\mathcal G \subset \Prob[\R]$ is called an \emph{affine class} of $Q \in \mathcal G$ if for each $\widetilde{Q} \in \mathcal G$ there exist $a \in \R$, $a \ne 0$, and $b \in \R$ such that $Y \sim \widetilde{Q}$ has the same distribution as the random variable $a\,X + b$ with $X \sim Q$. We write $\widetilde{Q} = \aff(Q, a, b)$. 

\begin{theorem} \label{theorem: one-dimensional depth}
A function $\Du \colon \R \times \Prob[\R] \to [0,1]$ verifies all~\ref{D1}--\ref{D4} if and only if for each affine class $\mathcal G$ and $Q \in \mathcal G$ as above there exists a function $\xi \colon \R \to [0,1]$ such that 
    \begin{itemize}
        \item $\xi$ is maximized at a median $m \in \median{Q} \in \R$ of $Q$, it is non-decreasing on $(-\infty, m]$ and non-increasing on $[m,\infty)$, and \[ \lim_{x \to -\infty} \xi(x) = \lim_{x \to \infty} \xi(x) = 0;\] 
        \item $\Du(x; Q) = \xi(x)$ for each $x \in \R$, and
        \item for each $\widetilde{Q} = \aff(Q, a, b) \in \mathcal G$ we have
    $\Du(x; \widetilde{Q}) = \Du\left( \frac{x - b}{a}; Q\right)$ for each $x \in \R$.
    \end{itemize}
\end{theorem}

\begin{proof}
Straightforward and omitted.
\end{proof}

The common halfspace and simplicial depths in $\R$ are obtained by relating $\xi$ to the distribution function $F$ of $Q$ as in~\eqref{eq: halfspace depth} and \eqref{eq: simplicial depth}, respectively. The properties of one-dimensional depths $\Du$ determine the traits of the resulting functional boxplot. 

\begin{theorem} \label{theorem: properties of infimal depth}
Let $\Du$ be a one-dimensional depth satisfying all~\ref{D1}--\ref{D4}. Consider the infimal depth given by
    $D(x; P) = \inf_{t \in [0,1]} \Du(x(t); P_t)$.
Then, the following is true.
    \begin{enumerate}[label=(\roman*), ref=(\roman*)]
    \item $D$ satisfies~\ref{FD1};
    \item Suppose that there exists $c \in (0,1]$ such that for each $Q \in \Prob[\R]$
        \begin{equation}    \label{eq: median condition}
        \Du(u; Q) \geq c \mbox{ if and only if }u \in \median{Q}. \tag{\textsf{\textcolor{red}{MD}}}
        \end{equation}
   Consider any $P \in \Prob[\FS]$ such that there exists $m \in \FS$ satisfying $m(t) \in \median{P_t}$ for each $t \in [0,1]$. Then, for this $P$, the depth $D$ satisfies~\ref{FD2}. In~\eqref{eq: function maximality}, one can then write $D(m; P) \geq c$. 
    \item $D$ satisfies~\ref{BC} and~\ref{FD3}.
    \end{enumerate}
\end{theorem}

The proof of Theorem~\ref{theorem: properties of infimal depth} can be found in the Supplementary Material. Condition~\eqref{eq: median condition} is true for the one-dimensional halfspace depth with $c = 1/2$, as
    \[  \DuH(u; Q) = \min\left\{ F(u), 1 - \lim_{v \to u-} F(v) \right\} \geq 1/2    \]
is true if and only if $F(u) \geq 1/2$ and $\lim_{v \to u-} F(v) \leq 1/2$, i.e. precisely when $u \in \median{Q}$. On the other hand, it is interesting to observe that~\eqref{eq: median condition} is not always true for the one-dimensional simplicial depth; an example is presented in the Supplementary Material. In the generality of Theorem~\ref{theorem: properties of infimal depth}, it is not true that $D$ must always obey~\ref{FD4}; this is shown in another example given in the Supplementary Material. For an infimal depth $D$ based on the one-dimensional halfspace~\eqref{eq: halfspace depth} or simplicial depth~\eqref{eq: simplicial depth}, condition~\ref{FD4} is however true for $D$, as shown in, e.g., \citet[Section~5]{Mosler_Polyakova2016} or \citet[Section~3.2]{Narisetty_Nair2016}.

\section{Simulations}   \label{section: simulations}

We now present a simulation study with centered Gaussian processes with covariance function
    \[  \Sigma(s, t) = \exp\left(- \frac{(s-t)^2}{h} \right) \quad \mbox{for }s, t \in [0,1].   \]
Random samples $X_1, \dots, X_n$ of size $n$ were generated for every experiment. Different values of $h \in (0,\infty)$ were chosen to control the volatility of functions. Small values of $h$ (that is, $\log(h) = -4$, i.e., $h \approx 0.018$) represent Gaussian processes that are quite rough and ``high-dimensional''. Larger values of $h$ (e.g., $\log(h) = 2$, that is $h \approx 7.389$) generate random functions that are quite stable in their course, and do not wiggle much.

For each sample, we considered functional boxplots based on the halfspace depth~\eqref{eq: halfspace depth} of both \begin{enumerate*}[label=(\roman*)] \item infimal type from~\eqref{eq: infimal depth}, and \item integrated type from~\eqref{eq: integrated depth}.\end{enumerate*} For the boxplots based on the infimal depth, two additional refinements of the ranking of the curves are considered: \begin{enumerate*}[label=(\roman*)] \item \textsf{erl} --- the extremal rank length \citep{Myllymaki_Mrkvicka2019}, which is equivalent with the tie-breaking procedure involved in the extremal depth for functional data \citep{Narisetty_Nair2016}, and \item \textsf{area} --- an interpolation technique based on an \textsf{area} index \citep{Myllymaki_Mrkvicka2019}. \end{enumerate*} In each sample, several characteristics of the boxplot were considered:
    \begin{itemize}
        \item Percentage of the sample curves that are contained in the central region of $50\%$ deepest sample functions (Table~\ref{Tab:boxs}). 
        \item Percentage of the sample curves that are contained inside the whiskers band (Table~\ref{Tab:whis}). 
        \item The mean width of the central region over $t \in [0,1]$ (Table~\ref{Tab:meaw}).
    \end{itemize}
Each of these characteristics was calculated in $100$ independent runs.


\begin{table}[htpb]
\centering
\caption{Means and standard deviations (in brackets) of the percentage of functions in the central region from the original dataset.} 
\label{Tab:boxs}
\resizebox{\linewidth}{!}{\begin{tabular}{ll|cccc}
   \hline
 depth & $n$ & $\log(h) = -4$ & $\log(h) = -2$ & $\log(h) = 0$ & $\log(h) = 2$ \\
   \hline
\multirow{ 3}{*}{\textsf{integrated}} & 50   & 51.000  \scriptsize{(1.694)}  & 51.540  \scriptsize{(1.925)}  & 50.700  \scriptsize{(1.219)}  & 50.040  \scriptsize{(0.281)}  \\ 
                                & 500  & 75.136  \scriptsize{(3.048)}  & 67.394  \scriptsize{(2.638)}  & 54.790  \scriptsize{(1.327)}  & 50.574  \scriptsize{(0.451)}  \\ 
                                & 5000 & 94.454  \scriptsize{(0.835)}  & 82.082  \scriptsize{(1.548)}  & 59.269  \scriptsize{(0.725)}  & 51.868  \scriptsize{(0.464)}  \\ \cdashline{1-6}
  \multirow{ 3}{*}{\textsf{infimal}}    & 50   & 60.940  \scriptsize{(13.062)} & 55.000  \scriptsize{(4.535)}  & 53.360  \scriptsize{(3.086)}  & 52.300  \scriptsize{(2.245)}  \\ 
                                & 500  & 51.046  \scriptsize{(0.689)}  & 50.494  \scriptsize{(0.400)}  & 50.234  \scriptsize{(0.231)}  & 50.202  \scriptsize{(0.221)}  \\ 
                                & 5000 & 50.122  \scriptsize{(0.080)}  & 50.051  \scriptsize{(0.040)}  & 50.029  \scriptsize{(0.027)}  & 50.021  \scriptsize{(0.022)}  \\ \cdashline{1-6}
  \multirow{ 3}{*}{\textsf{erl}}        & 50   & 50.000  \scriptsize{(0.000)}  & 50.000  \scriptsize{(0.000)}  & 50.000  \scriptsize{(0.000)}  & 50.000  \scriptsize{(0.000)}  \\ 
                                & 500  & 50.000  \scriptsize{(0.000)}  & 50.000  \scriptsize{(0.000)}  & 50.000  \scriptsize{(0.000)}  & 50.000  \scriptsize{(0.000)}  \\ 
                                & 5000 & 50.000  \scriptsize{(0.000)}  & 50.000  \scriptsize{(0.000)}  & 50.000  \scriptsize{(0.000)}  & 50.000  \scriptsize{(0.000)}  \\ \cdashline{1-6}
  \multirow{ 3}{*}{\textsf{area}}       & 50   & 50.000  \scriptsize{(0.000)}  & 50.000  \scriptsize{(0.000)}  & 50.000  \scriptsize{(0.000)}  & 50.000  \scriptsize{(0.000)}  \\ 
                                & 500  & 50.000  \scriptsize{(0.000)}  & 50.000  \scriptsize{(0.000)}  & 50.000  \scriptsize{(0.000)}  & 50.000  \scriptsize{(0.000)}  \\ 
                                & 5000 & 50.000  \scriptsize{(0.000)}  & 50.000  \scriptsize{(0.000)}  & 50.000  \scriptsize{(0.000)}  & 50.000  \scriptsize{(0.000)}  \\ 
   \hline
\end{tabular}}
\end{table}

\begin{table}[htpb]
\centering
\caption{Means and standard deviations (in brackets) of the percentage of functions within the whiskers band from the original dataset.} 
\label{Tab:whis}
\resizebox{\linewidth}{!}{\begin{tabular}{ll|cccc}
   \hline
 depth & $n$ & $\log(h) = -4$ & $\log(h) = -2$ & $\log(h) = 0$ & $\log(h) = 2$ \\
   \hline
\multirow{ 3}{*}{\textsf{integrated}} & 50   & 100.000  \scriptsize{(0.000)} & 99.880  \scriptsize{(0.556)}  & 98.900  \scriptsize{(1.667)}  & 97.620  \scriptsize{(3.250)}  \\ 
                                & 500  & 100.000  \scriptsize{(0.000)} & 100.000  \scriptsize{(0.000)} & 99.878  \scriptsize{(0.165)}  & 99.282  \scriptsize{(0.452)}  \\ 
                                & 5000 & 100.000  \scriptsize{(0.000)} & 100.000  \scriptsize{(0.000)} & 99.975  \scriptsize{(0.023)}  & 99.518  \scriptsize{(0.100)}  \\ \cdashline{1-6}
  \multirow{ 3}{*}{\textsf{infimal}}    & 50   & 99.980  \scriptsize{(0.200)}  & 100.000  \scriptsize{(0.000)} & 99.640  \scriptsize{(0.871)}  & 99.120  \scriptsize{(1.565)}  \\ 
                                & 500  & 100.000  \scriptsize{(0.000)} & 100.000  \scriptsize{(0.000)} & 99.986  \scriptsize{(0.051)}  & 99.780  \scriptsize{(0.210)}  \\ 
                                & 5000 & 100.000  \scriptsize{(0.000)} & 100.000  \scriptsize{(0.000)} & 99.995  \scriptsize{(0.010)}  & 99.879  \scriptsize{(0.055)}  \\ \cdashline{1-6}
  \multirow{ 3}{*}{\textsf{erl}}        & 50   & 99.980  \scriptsize{(0.200)}  & 100.000  \scriptsize{(0.000)} & 99.480  \scriptsize{(1.087)}  & 98.840  \scriptsize{(1.824)}  \\ 
                                & 500  & 100.000  \scriptsize{(0.000)} & 100.000  \scriptsize{(0.000)} & 99.986  \scriptsize{(0.051)}  & 99.766  \scriptsize{(0.218)}  \\ 
                                & 5000 & 100.000  \scriptsize{(0.000)} & 100.000  \scriptsize{(0.000)} & 99.995  \scriptsize{(0.010)}  & 99.878  \scriptsize{(0.055)}  \\ \cdashline{1-6}
  \multirow{ 3}{*}{\textsf{area}}       & 50   & 99.980  \scriptsize{(0.200)}  & 100.000  \scriptsize{(0.000)} & 99.480  \scriptsize{(1.087)}  & 98.840  \scriptsize{(1.824)}  \\ 
                                & 500  & 100.000  \scriptsize{(0.000)} & 100.000  \scriptsize{(0.000)} & 99.986  \scriptsize{(0.051)}  & 99.774  \scriptsize{(0.212)}  \\ 
                                & 5000 & 100.000  \scriptsize{(0.000)} & 100.000  \scriptsize{(0.000)} & 99.995  \scriptsize{(0.010)}  & 99.879  \scriptsize{(0.055)}  \\ 
   \hline
\end{tabular}}
\end{table}

\begin{table}[htpb]
\centering
\caption{Means and standard deviations (in brackets) of the mean central region width.} 
\label{Tab:meaw}
\resizebox{\linewidth}{!}{\begin{tabular}{ll|cccc}
   \hline
 depth & $n$ & $\log(h) = -4$ & $\log(h) = -2$ & $\log(h) = 0$ & $\log(h) = 2$ \\
   \hline
\multirow{ 3}{*}{\textsf{integrated}} & 50   & 3.184  \scriptsize{(0.223)} & 2.565  \scriptsize{(0.237)} & 1.860  \scriptsize{(0.229)} & 1.431  \scriptsize{(0.195)} \\ 
                                & 500  & 4.733  \scriptsize{(0.178)} & 3.632  \scriptsize{(0.173)} & 2.360  \scriptsize{(0.104)} & 1.774  \scriptsize{(0.086)} \\ 
                                & 5000 & 6.067  \scriptsize{(0.133)} & 4.515  \scriptsize{(0.131)} & 2.714  \scriptsize{(0.066)} & 2.050  \scriptsize{(0.058)} \\ \cdashline{1-6}
  \multirow{ 3}{*}{\textsf{infimal}}    & 50   & 3.273  \scriptsize{(0.408)} & 2.559  \scriptsize{(0.247)} & 1.950  \scriptsize{(0.224)} & 1.514  \scriptsize{(0.194)} \\ 
                                & 500  & 3.587  \scriptsize{(0.065)} & 2.813  \scriptsize{(0.068)} & 2.117  \scriptsize{(0.075)} & 1.646  \scriptsize{(0.077)} \\ 
                                & 5000 & 3.723  \scriptsize{(0.020)} & 2.892  \scriptsize{(0.022)} & 2.184  \scriptsize{(0.026)} & 1.714  \scriptsize{(0.022)} \\ \cdashline{1-6}
  \multirow{ 3}{*}{\textsf{erl}}        & 50   & 2.959  \scriptsize{(0.173)} & 2.413  \scriptsize{(0.189)} & 1.852  \scriptsize{(0.197)} & 1.448  \scriptsize{(0.186)} \\ 
                                & 500  & 3.555  \scriptsize{(0.061)} & 2.797  \scriptsize{(0.066)} & 2.108  \scriptsize{(0.075)} & 1.638  \scriptsize{(0.076)} \\ 
                                & 5000 & 3.720  \scriptsize{(0.019)} & 2.890  \scriptsize{(0.022)} & 2.183  \scriptsize{(0.026)} & 1.713  \scriptsize{(0.022)} \\ \cdashline{1-6}
  \multirow{ 3}{*}{\textsf{area}}       & 50   & 2.973  \scriptsize{(0.181)} & 2.414  \scriptsize{(0.194)} & 1.855  \scriptsize{(0.198)} & 1.448  \scriptsize{(0.184)} \\ 
                                & 500  & 3.560  \scriptsize{(0.061)} & 2.800  \scriptsize{(0.065)} & 2.110  \scriptsize{(0.074)} & 1.640  \scriptsize{(0.075)} \\ 
                                & 5000 & 3.720  \scriptsize{(0.019)} & 2.890  \scriptsize{(0.022)} & 2.183  \scriptsize{(0.026)} & 1.714  \scriptsize{(0.022)} \\ 
   \hline
\end{tabular}}
\end{table}

Table~\ref{Tab:boxs} illustrates the degree of how much the employed depths violate the band convexity criterion in practice; while ideally, the mean in Table~\ref{Tab:boxs} should be around $50\%$, we see that especially for $n$ large and volatile datasets, the central region based on integrated depths may contain much higher number of functions than desired. In Table~\ref{Tab:boxs}, we also recognize the amount of depth-ties for functions defining the edge of the 50\% central region for the pure infimal depth (\textsf{infimal}) without tie-breaking. When these ties are broken using \textsf{erl} or \textsf{area}, these problems disappear. Table~\ref{Tab:whis} shows that the integrated boxplot is too narrow for smooth functions (i.e., the coverage does not attain values near $100\%$ for larger $h$), while Table~\ref{Tab:meaw} shows that the integrated boxplot constructed from a higher number of functions tends to be too wide for volatile data. It also shows that the width of the integrated $50\%$ central region quickly increases with the number of functions. This observation complements the statement of our Theorem~\ref{theorem: infinite band}. The Supplementary Material presents several further characteristics observed in the simulation study. 

In addition to the previous results, in each run of our simulation study, a second random sample $Y_1, \dots, Y_{n_{\mathrm{test}}}$ of $n_{\mathrm{test}} = 10~000$ random functions from the same model was drawn, independently of the original data. For a boxplot based on the original data $X_1, \dots, X_n$, we used the testing sample to assess the percentage of the test curves $Y_1, \dots, Y_{n_{\mathrm{test}}}$ contained inside the central region based on $X_1, \dots, X_n$ (Table~\ref{Tab:boxt}). Here, one can not realistically expect that the coverage will be exactly $50\%$ since the sample central regions estimate the population central regions. However, one should require that the coverage of sample central regions converges from below to $50\%$ as $n \to \infty$. Table~\ref{Tab:boxt} thus demonstrates the convergence speed of the sample central regions to their population counterparts. The coverage of all infimal depths converges from below to $50\%$, whereas the coverage of the integrated depth seems to converge to a number that is higher than $50\%$, and that depends on the wiggliness of the functions.  The central regions of the refined infimal depths (\textsf{erl} and \textsf{area}) are thinner due to the tie-breaking nature of those depths. Their coverage is, therefore, necessarily smaller than those for the pure infimal depth. 
 
\begin{table}[htpb]
\centering
\caption{Means and standard deviations (in brackets) of the percentage of functions in the central region from the test sample.} 
\label{Tab:boxt}
\resizebox{\linewidth}{!}{\begin{tabular}{ll|cccc}
   \hline
 depth & $n$ & $\log(h) = -4$ & $\log(h) = -2$ & $\log(h) = 0$ & $\log(h) = 2$ \\
   \hline
\multirow{ 3}{*}{\textsf{integrated}} & 50   & 14.138  \scriptsize{(4.634)}  & 28.422  \scriptsize{(6.257)}  & 36.343  \scriptsize{(6.911)}  & 38.832  \scriptsize{(6.635)}  \\ 
                                & 500  & 68.254  \scriptsize{(4.114)}  & 63.928  \scriptsize{(3.310)}  & 52.062  \scriptsize{(2.597)}  & 48.392  \scriptsize{(2.301)}  \\ 
                                & 5000 & 93.532  \scriptsize{(0.930)}  & 81.649  \scriptsize{(1.665)}  & 58.917  \scriptsize{(1.065)}  & 51.601  \scriptsize{(0.909)}  \\ \cdashline{1-6}
  \multirow{ 3}{*}{\textsf{infimal}}    & 50   & 21.412  \scriptsize{(12.173)} & 31.714  \scriptsize{(7.771)}  & 39.567  \scriptsize{(7.158)}  & 42.064  \scriptsize{(6.420)}  \\ 
                                & 500  & 40.620  \scriptsize{(2.405)}  & 45.369  \scriptsize{(2.337)}  & 47.405  \scriptsize{(2.469)}  & 48.060  \scriptsize{(2.352)}  \\ 
                                & 5000 & 48.151  \scriptsize{(0.737)}  & 49.052  \scriptsize{(0.845)}  & 49.436  \scriptsize{(0.964)}  & 49.738  \scriptsize{(0.778)}  \\ \cdashline{1-6}
  \multirow{ 3}{*}{\textsf{erl}}        & 50   & 13.176  \scriptsize{(3.909)}  & 27.533  \scriptsize{(5.971)}  & 36.480  \scriptsize{(6.511)}  & 39.969  \scriptsize{(6.216)}  \\ 
                                & 500  & 39.517  \scriptsize{(2.230)}  & 44.843  \scriptsize{(2.228)}  & 47.150  \scriptsize{(2.469)}  & 47.850  \scriptsize{(2.345)}  \\ 
                                & 5000 & 48.015  \scriptsize{(0.726)}  & 48.994  \scriptsize{(0.856)}  & 49.404  \scriptsize{(0.964)}  & 49.710  \scriptsize{(0.781)}  \\ \cdashline{1-6}
  \multirow{ 3}{*}{\textsf{area}}       & 50   & 13.545  \scriptsize{(4.247)}  & 27.519  \scriptsize{(6.057)}  & 36.490  \scriptsize{(6.573)}  & 39.904  \scriptsize{(6.214)}  \\ 
                                & 500  & 39.609  \scriptsize{(2.246)}  & 44.906  \scriptsize{(2.222)}  & 47.174  \scriptsize{(2.476)}  & 47.880  \scriptsize{(2.342)}  \\ 
                                & 5000 & 48.022  \scriptsize{(0.728)}  & 48.990  \scriptsize{(0.854)}  & 49.407  \scriptsize{(0.967)}  & 49.713  \scriptsize{(0.784)}  \\ 
   \hline
\end{tabular}}
\end{table}

\section{Conclusions and discussion}    \label{section: conclusions}

We have shown that the functional boxplots founded on local principles must be based on depths that follow the same paradigm. Since the functional boxplot decides about the outlyingness of functions locally (by evaluating whether a function falls outside the whiskers band at some argument $t \in [0,1]$), the central region that defines this threshold must also be based on local principles. We demonstrated that the only functional depths satisfying the ensuing desirable band convexity property are depths of infimal type. 

From a practical perspective, we observed that depths of infimal type tend to produce ties in rankings. We recommend refining the depth using a tie-breaking procedure before constructing functional boxplots. Both refinements investigated in this paper, \textsf{erl} and \textsf{area}, have demonstrated to be viable by giving near-desired coverage of the sample central region in our simulation study. Our simulation study also shows that their sample central region width is quite stable; therefore, such refinements can be used to approximate the population central region of the functional distribution. The situation is far less favorable for integrated depths, as demonstrated theoretically and in simulation studies. 

Reaching beyond classical functional boxplots, it is natural to consider functional data with multi-dimensional domains and/or co-domains \citep{Genton_etal2014, Yao_etal2020, Qu_Genton2022}. The situation with a compact domain $\mathcal T \subset \R^d$ and $\R$-valued responses is simple and follows the theory derived in this paper without any alteration. The case of vector-valued functions $x \colon \mathcal T \to \R^m$ is less trivial. First, we can again assume that $\mathcal T = [0,1]$, without loss of generality. There are two natural ways to extend the functional band to vector-valued functions. The first possibility is to require convexity of all the slices of bands in $\FS$. Let $B \subseteq \FS$ be a subset of the space of functions $\left\{ x \colon [0,1] \to \R^m \right\}$. The \emph{band} of $B$ is then defined as
    \begin{equation}    \label{eq: convex band}
    \begin{aligned}
    \band{B} & = \left\{ x \in \FS \colon x(t) \in   \right. \\ 
    & \left. \conv{ \left\{ y(t) \colon y \in B \right\}} \mbox{ for all }t \in [0,1] \right\}, 
    \end{aligned}
    \end{equation}   
where $\conv{A}$ is the convex hull of $A \subset \R^m$. Another option is to consider (finite-dimensional) bands themselves as a generalization of the intervals used in~\eqref{eq: band}. This leads to the alternative notion of a band
    \begin{equation}    \label{eq: rectangular band}
    \begin{aligned}
    \band{B} & = \Big\{ x \in \FS \colon \inf_{y \in B} y_j(t) \leq x_j(t) \leq \sup_{y \in B} y_j(t) \\
    & \mbox{ for all }t \in [0,1] \mbox{ and }j = 1, \dots, m \Big\}. 
    \end{aligned}
    \end{equation}    
Here we use the notation $y = (y_1, \dots, y_m)\tr$ for any $y \in \FS$, where $y_j \colon [0,1] \to \R$ are the component functions of $y$. The advantage of convex bands~\eqref{eq: convex band} is their geometric interpretation; they are also conveniently visualized if $m = 2$. The appeal of rectangular bands~\eqref{eq: rectangular band} comes from their simplicity. Each slice of $\band{B}$ from~\eqref{eq: rectangular band} at $t \in [0,1]$ takes the form of an $m$-dimensional rectangle $[a_1(t), b_1(t)] \times \dots \times [a_m(t), b_m(t)] \subseteq \R^m$, where $a_j(t) \leq b_j(t)$ for each $j = 1, \dots, m$ and $t \in [0,1]$. This makes it easy to visualize the band for $m > 2$, as each its component is determined only by two functions $a_j, b_j \colon [0,1] \to \Rinf$. The coordinate-wise approach~\eqref{eq: rectangular band} necessarily loses information about the dependence structure between the component functions. Each of the two methods~\eqref{eq: convex band} and~\eqref{eq: rectangular band} has its particular pros and cons. While some of our findings can also be adapted to multivariate functional data (e.g., depths of infimal type are preferable to those of integrated type), the complexity of visualization of multivariate curves invariably calls for further study of boxplots for multivariate functional data.

Concluding from a wider perspective, this paper demonstrated that infimal depths are the best choice for constructing the central region and functional boxplot. This fact, however, certainly does not diminish the value of other functional depths, which can be useful for ranking the functional data, and nonparametric inference.

If our task is to detect functional outliers, the functional boxplot with the preferred infimal depth is useful for detecting locally outlying functions. To detect global outliers or functions outlying in shape, the infimal depth-based boxplot may be used in conjunction with appropriate functional transformations revealing the outliers \citep{Dai_etal2020}.

\def\polhk#1{\setbox0=\hbox{#1}{\ooalign{\hidewidth
  \lower1.5ex\hbox{`}\hidewidth\crcr\unhbox0}}}


\end{document}